\documentclass{amsart}

\usepackage[all]{xy}
\theoremstyle{definition}
\newtheorem{theorem}{Theorem}[section]

\theoremstyle{remark}
\newtheorem{rem}[theorem]{Remark}
\newtheorem*{acknowledgment}{Acknowledgments}

\DeclareMathOperator{\C}{\mathbb{C}}
\DeclareMathOperator{\Q}{\mathbb{Q}}
\DeclareMathOperator{\N}{\mathbb{N}}
\DeclareMathOperator{\Z}{\mathbb{Z}} 

\newcommand{\seq}[1]{{\boldsymbol{#1}}}
 
\DeclareMathOperator{\id}{id} 
\DeclareMathOperator{\sgn}{sgn}

\title{Quiver mutation loops and partition $q$-series}

\author{Akishi~Kato}%
\address{Graduate School of Mathematical Sciences, 
The University of Tokyo, 
3-8-1 Komaba, Meguro-ku, Tokyo 153-8914, Japan.}
\email{akishi@ms.u-tokyo.ac.jp}

\author{Yuji~Terashima}
\address{Graduate School of Information Science and Engineering, 
Tokyo Institute of Technology, 
2-12-1 Ookayama, Meguro-ku, Tokyo 152-8550, Japan.}
\email{tera@is.titech.ac.jp}

\numberwithin{equation}{section}

\begin{document}

\begin{abstract}
A quiver mutation loop is a sequence of mutations and vertex
relabelings, along which a quiver transforms back to the original
form. For a given mutation loop $\gamma$, we introduce a quantity called
a \emph{partition $q$-series} $Z(\gamma)$ which takes values in
$\N[[q^{1/\Delta}]]$ where $\Delta$ is some positive integer. The
partition $q$-series are invariant under pentagon moves. If the quivers
are of Dynkin type or square products thereof, they reproduce so-called
fermionic or quasi-particle character formulas of certain modules
associated with affine Lie algebras. They enjoy nice modular properties
as expected from the conformal field theory point of view.
\end{abstract}

\maketitle

\section{Introduction}

Quiver mutations are now ubiquitous in many branches of mathematics and
mathematical physics, such as Donaldson-Thomas theory, low dimensional
topology, representation theory, quantum field theories. Quiver
mutations are now recognized as important tools, along with cluster
algebras.

The main purpose of this paper is to introduce quantities called
\emph{partition $q$-series} directly at the level of quiver mutation
sequences. The definition requires only combinatorial data of quivers
and mutation sequences, and completely independent of the details of the
problem. In fact, one motivation is to provide a solid mathematical
foundation to extract essential information of the partition function
of a $3$-dimensional gauge theory associated with a sequence of quiver
mutations which is introduced in \cite{TY}.  It is hoped that a deeper
understanding of the partition $q$-series will help uncover the hidden
combinatorial structure and shed new light on the mystery of
quantization.

A quiver mutation loop is a sequence of mutations and vertex
relabelings along which a quiver transforms back to the original
form. For a given mutation loop $\gamma$, we associate a quantity called
a \emph{partition $q$-series} $Z(\gamma)$ which takes values in
$\N[[q^{1/\Delta}]]$ where $\Delta$ is some positive integer.  The
partition $q$-series are closely related to the quantum dilogarithms,
and satisfy various invariance properties such as pentagon relations. If
the quivers are of Dynkin type or square products thereof, they
reproduce so-called fermionic or quasi-particle character formulas
of certain modules associated with affine Lie algebras. They enjoy nice
modular properties as expected from the conformal field theory point of
view.

The paper is organized as follows. In Section \ref{sec:BG}, we recall
the basic definitions of quiver mutations. In Section \ref{sec:Zdef}, we
introduce the partition $q$-series $Z(\gamma)$ for the mutation loop
$\gamma$. Since the definition is a slightly complicated, we supplied a
few simple examples. In Section \ref{sec:Invariance}, we introduce the
notion of ``pentagon move'' of mutation loops, and show that the
partition $q$-series are invariant under such moves. In the final
Sections \ref{sec:ZI} and \ref{sec:ZII}, we treat the quivers of
simply-laced Dynkin type or square products thereof. It is demonstrated
that if we choose a special mutation loop, the associated partition
$q$-series are nothing but the ``fermionic character formulas'' of
certain modules associated with affine Lie algebras. Up to
multiplication by appropriate powers of $q$, they are conjectured to be
modular forms with respect to a certain congruence subgroup of
$SL(2,\Z)$, as expected from the conformal field theory point of view.

The paper \cite{Cecotti2010} proposes a relation between four-dimensional
gauge theories and parafermionic conformal field theories. In particular,
they claim that the $L^2$-trace of the half monodromy is written in
terms of characters. It would be interesting to find a precise relation
with their work.
\begin{acknowledgment}
We would like to thank H. Fuji, K. Hikami, T. Kitayama, A. Kuniba,
K. Nagao, J. Suzuki, S. Terashima and M. Yamazaki for helpful
discussion. This work was partially supported by Japan Society for the
Promotion of Science (JSPS), Grants-in-Aid for Scientific Research Grant
(KAKENHI) Number 23654079 and 25400083.
\end{acknowledgment}

\section{Backgrounds} 
\label{sec:BG}

\subsection{Quivers and mutations} 

A \emph{quiver} $Q$ is an oriented graph given by a set of vertices
$Q_0$, a set of arrows $Q_1$ and two maps ``source'' $s: Q_1 \to Q_0$
and ``target'' $t: Q_1 \to Q_0$. A quiver $Q$ is \emph{finite} if the
sets $Q_0$ and $Q_1$ are finite.  Throughout this paper, we will assume
all quivers are finite, and an isomorphism
$Q_{0}\stackrel{\sim}{\rightarrow} \{1,\dots,n\}$, called
\emph{labeling}, is fixed.

Let $Q$ be a quiver.  A \emph{loop} or \emph{$1$-cycle} of $Q$ is an
arrow $\alpha$ whose source and target coincide. A \emph{$2$-cycle} of
$Q$ is a pair of distinct arrows $\beta$ and $\gamma$ such that
$s(\beta)=t(\gamma)$ and $t(\beta)=s(\gamma)$:
\begin{equation*}
 \text{loop} \quad \vcenter{
  \xymatrix @R=8mm @C=8mm
  @M=2pt{\bullet \ar@(ur,dr)[]}}  \qquad \qquad
  \text{2-cycle}\quad \vcenter{ \xymatrix @R=8mm @C=8mm @M=2pt{\bullet
  \ar@/^/[r] & \bullet \ar@/^/[l]} }
\end{equation*}
In this paper, we treat quivers without loops or 2-cycles.  For a
quiver $Q$, $Q^{op}$ denotes the quiver obtained from $Q$ by reversing
all arrows.

A \emph{Dynkin quiver} is a quiver $Q$ whose underlying graph
$\underline{Q}$, a graph obtained by forgetting the orientation of
arrows, is a Dynkin diagram.

A vertex $i$ of a quiver is a \emph{source} (respectively, a
\emph{sink}) if there are no arrows $\alpha$ with target $i$
(respectively, with source $i$).  A quiver is \emph{alternating} if each
of its vertices is a source or a sink. For an alternating graph $Q$, the
\emph{sign} of the vertex $i$ is defined as $\sgn(i)=1$ if $i$ is a
source and $\sgn(i)=-1$ if $i$ is a sink. Here are some examples of
alternating Dynkin quivers:
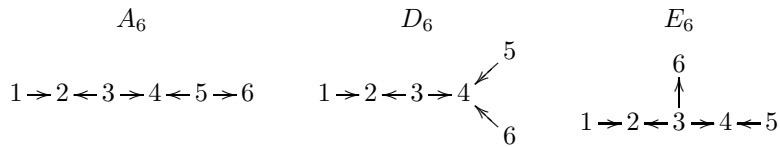
\begin{equation*}
 \begin{array}{ccccc}
  A_{6} && D_{6} && E_{6}\\
  \vcenter{\xymatrix @R=2mm @C=3mm @M=2pt{
   1 \ar[r]  &2 & 3 \ar[l]\ar[r] & 4 & 5  \ar[l]\ar[r] & 6
   }}
   && 
   \vcenter{\xymatrix @R=2mm @C=3mm @M=2pt{
   &&&& 5 \ar[dl] \\
  1 \ar[r]  &2 & 3 \ar[l]\ar[r] & 4  \\
  &&&&6 \ar[ul]  }}
   &&
   \vcenter{\xymatrix @R=4mm @C=3mm @M=2pt{
   && 6 \\
  1 \ar[r]  &2 & 3 \ar[l]\ar[r]\ar[u] & 4 & 5 \ar[l]
   }}
 \end{array}
\end{equation*}

For a quiver $Q$ and its vertex $k$, the \emph{mutated quiver}
$\mu_k(Q)$ is defined \cite{Fomin2002}: it has the same set of vertices
as $Q$; its set of arrows is obtained from that of $Q$ as follows:
\begin{itemize}
 \item[1)] for each path $i\to k\to j$ of length two,
	 add a new arrow 
	 $i\to j$;
 \item[2)] reverse all arrows with source or target $k$;
 \item[3)] remove the arrows in a maximal set of pairwise
disjoint $2$-cycles.
\end{itemize}
The following two quivers are obtained from each other by mutating at
the black vertex
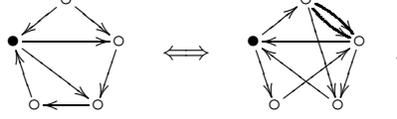
\begin{equation*}
\vcenter{\begin{xy} 0;<2pt,0pt>:<0pt,-2pt>:: 
(20,8) *+<2pt>{\circ} ="0",
(4,20) *+<2pt>{\circ} ="1",
(16,20) *+<2pt>{\circ} ="2",
(0,8) *+<2pt>{\bullet} ="3",
(10,0) *+<2pt>{\circ} ="4",
"0", {\ar"2"},
"3", {\ar"0"},
"4", {\ar"0"},
"2", {\ar"1"},
"1", {\ar"3"},
"3", {\ar"2"},
"4", {\ar"3"},
\end{xy}}
\quad
\Longleftrightarrow
\quad
\vcenter{
\begin{xy} 0;<2pt,0pt>:<0pt,-2pt>:: 
(20,8) *+<2pt>{\circ} ="0",
(4,20) *+<2pt>{\circ} ="1",
(16,20) *+<2pt>{\circ} ="2",
(0,8) *+<2pt>{\bullet} ="3",
(10,0) *+<2pt>{\circ} ="4",
"1", {\ar"0"},
"0", {\ar"2"},
"0", {\ar"3"},
"4":{\ar@/^2pt/"0"},
"4":{\ar@/_2pt/"0"},
"3", {\ar"1"},
"2", {\ar"3"},
"4", {\ar"2"},
"3", {\ar"4"},
\end{xy}}\quad .
\end{equation*}
There is a bijection
\begin{equation*}
\left\{
 \text{
 \parbox[c]{.35\textwidth}{the quivers without loops or $2$-cycles, 
 $Q_{0}= \{1, \ldots, n\}$}
}
\right\}
\longleftrightarrow
\left\{
 \text{
 \parbox[c]{.35\textwidth}{the skew-symmetric integer $n\times
 n$-matrices $B=(b_{ij})$}
}
\right\}
\end{equation*}
by 
\begin{equation}
b_{ij}=\#\{(i {\to} j) \in Q_{1}\}-\#\{(j {\to} i) \in Q_{1}\}.
\end{equation}
The above operation of quiver mutation corresponds to matrix mutation
defined by Fomin-Zelevinsky \cite{Fomin2002}.  The matrix $B'$
corresponding to $\mu_k(Q)$ is given by \cite{Fomin2007}
\begin{equation} \label{eq:matrix-mutation}
b'_{ij} =\left\{ \begin{array}{ll} -b_{ij} & \text{if $i=k$ or $j=k$} \\
b_{ij}+\sgn(b_{ik}) \max(b_{ik}b_{kj},0) & \text{otherwise.}\end{array} \right.
\end{equation}

\subsection{Mutation sequences and mutation loops}
\label{sec:mu-seq-loop}

A finite sequence of vertices of $Q$, $\seq{m}=(m_{1},m_{2},\dots,m_{T})$
is called \emph{mutation sequence}. This can be regarded as a (discrete)
time evolution of quivers:
\begin{equation*}
 Q(0):=Q,\qquad Q(t) := \mu_{m_{t}}(Q(t-1)),  \qquad (1\leq t\leq T).
\end{equation*}
$Q(0)$ and $Q(T)$ are called the \emph{initial} and the \emph{final}
quiver, respectively. 

Suppose further that $Q(0)$ and $Q(T)$ are isomorphic,
namely, the composed mutation
\begin{equation}
 \mu_{\seq{m}}:= \mu_{m_{T}}\circ \dots \circ \mu_{m_{2}} \circ \mu_{m_{1}}
\end{equation}
transforms $Q$ into a quiver isomorphic to $Q$. An isomorphism $\varphi
: Q(T)\to Q(0)$ regarded as a bijection on the set of vertices, is
called \emph{boundary condition} of the mutation sequence $\seq{m}$.
Using the fixed labeling $Q_{0}\stackrel{\sim}{\rightarrow}
\{1,\dots,n\}$, we represent $\varphi$ by an element in the symmetric
group $S_n$. The triple $\gamma=(Q;\seq{m},\varphi)$ is called a
\emph{mutation loop}.

\section{Partition $q$-series and their examples}
\label{sec:Zdef}

In this section, we introduce a quantity called partition $q$-series
$Z(\gamma)$ for a mutation loop $\gamma$, in the same spirit as
partition functions of statistical mechanics. Roughly speaking,
$Z(\gamma)$ is defined as a sum of weights over all possible states,
while the weights are expressed as a product of local factors. For
clarity's sake, sample computations of $Z(\gamma)$ are presented.

\subsection{Definition of partition $q$-series}

Let $\gamma=(Q;\seq{m},\varphi)$ be a mutation loop.
We first introduce a family of \emph{$s$-variables} $\{s_{i}\}$ and
\emph{$k$-variables} $\{k_{t}\}$ as follows.
\begin{itemize}
 \item[(i)] An ``initial'' $s$-variable $s_{v}$ is attached to each
	    vertex $v$ of the initial quiver $Q$.

 \item[(ii)] Every time we mutate at vertex $v$, we add a ``new''
	    $s$-variable associated with $v$. We often use $s_{v}$,
	    $s_{v}'$, $s_{v}''$, $\dots$ to distinguish $s$-variables
	    attached to the same vertex.

 \item[(iii)] We associate a $k$-variable $k_{t}$ with each mutation at
	    $m_{t}$.
	      
 \item[(iv)] If two vertices are related by a boundary condition, then
	    the corresponding $s$-variables are identified.
\end{itemize}
As we will soon see, the $s$- and $k$-variables are not considered
independent; we impose a linear relation for each mutation step. We also
define a \emph{weight} of each mutation as a function of these
variables.

Suppose that the quiver $Q(t-1)$ equipped with $s$-variables $\{s_{i}\}$
is mutated at vertex $v=m_{t}$ to give $Q(t)$. Then $k$- and
$s$-variables are required to satisfy 
\begin{equation}
 \label{eq:k-and-s}
 k_{t} = s_{v}+s'_{v} -\sum_{a\to v}s_{a}.
\end{equation}
Here, $s'_{v}$ is the ``new'' $s$-variable attached to mutated vertex
$v$, and the sum is over all the arrows of $Q(t-1)$ whose target vertex is $v$.

The \emph{weight of the mutation} $\mu_{m_{t}}:Q(t-1)\to Q(t)$ at
$v=m_{t}$ is defined as
\begin{equation}
 \label{eq:W-def}
  W(m_{t}):=
  \frac{q^{\frac{1}{2} 
   \left(s_{v}+s'_{v} -\sum_{a\to v}s_{a}\right)
   \left(s_{v}+s'_{v} -\sum_{v\to b}s_{b}\right)}
   }{(q)_{s_{v}+s'_{v} -\sum_{a\to v}s_{a}}}
   =
  \frac{q^{\frac{1}{2} k_{t}   \left(s_{v}+s'_{v} -\sum_{v\to b}s_{b}\right)}
   }{(q)_{k_{t}}},
\end{equation}
where 
\begin{equation}
  (x;q)_{n} := \prod_{k=0}^{n-1}(1-x q^{k}),\quad\qquad (q)_{n}:=(q;q)_{n}
\end{equation}
is the $q$-Pochhammer symbol. 

The \emph{weight of the mutation loop} $\gamma$ is then defined as the
product  over all mutations,
\begin{equation}
 W(\gamma) = \prod_{t=1}^{T} W(m_{t}).
\end{equation}
Clearly $W(\gamma)$ has a structure
\begin{equation}
 \label{eq:W-seq} W(\gamma) = \frac{q^{H(\seq{s})}}{
  \prod_{t=1}^{T} (q)_{k_{t}}} 
\end{equation}
where $H(\seq{s})$ is a quadratic form in $s$-variables. 

For example, suppose we mutate the following quiver at vertex
$v=m_{t}$. (All arrows not incident on $v$ are omitted.)
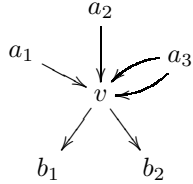
\begin{equation*}
 \xygraph{!{<0cm,0cm>;<20pt,0cm>:<0cm,20pt>::}
!{(-1.5,0.8) }*+{a_{1}}="a1"
!{(0,1.6) }*+{a_{2}}="a2"
!{(1.5,0.7)}*+{a_{3}}="a3"
!{(0,0) }*+{v}="v"
!{(-1,-1.4) }*+{b_{1}}="b1"
!{(1,-1.4) }*+{b_2}="b2"
"a1":"v"   "a2":"v"  "v":"b1"  "v":"b2"
"a3":@/_5pt/"v"   "a3":@/^5pt/"v" 
}
\end{equation*}
In this case, the relation \eqref{eq:k-and-s} reads 
\begin{equation*}
 k_{t}=s_{v}+s'_{v}-s_{a_1}-s_{a_2}-2s_{a_3}
\end{equation*}
and the weight of the mutation is 
\begin{equation*}
 W(m_{t})
  =\frac{q^{\frac{1}{2} 
  (s_{v}+s'_{v}-s_{a_1}-s_{a_2}-2s_{a_3})  
  (s_{v}+s'_{v}-s_{b_1}-s_{b_2})}}
  {(q)_{s_{v}+s'_{v}-s_{a_1}-s_{a_2}-2s_{a_3}}}
  =\frac{q^{\frac{1}{2} 
  k_{t}  
  (s_{v}+s'_{v}-s_{b_1}-s_{b_2})}}
  {(q)_{k_{t}}}.
\end{equation*}
Note that both the linear relations \eqref{eq:k-and-s} and the weight
\eqref{eq:W-def} uses only the local information around the mutating
vertex.

The relation \eqref{eq:k-and-s} allows us to express each $k$-variable
as a $\Z$-linear combination of $s$-variables. If these relations are
invertible as a whole, namely, if one can express each $s$-variable as a
$\Q$-linear combination of $k$-variables, then, the mutation loop
$\gamma$ is called \emph{nondegenerate}.\footnote{It is not easy to
decide whether or not $\gamma$ is nondegenerate, just looking ``local''
structure of $\gamma$. However, the following remark is in order.  For a
mutation loop with $T$ mutations, the number of $k$-variables is $T$ by
the rule (iii). The number of independent $s$-variables is also $T$; we
start with $\#Q_{0}$ of ``initial'' $s$-variables (i), add $T$ ``new''
ones (ii), but $\#Q_{0}$ of $s$-variables are identified via the
boundary condition $\varphi$ (iv). So there is a good chance of $\gamma$
being nondegenerate.} Suppose the mutation loop $\gamma$
is nondegenerate. Then the quadratic form $H(\seq{s})$ in
\eqref{eq:W-seq} can be expressed as a quadratic form
$F(\seq{k})$ in $k$-variables:
\begin{equation}
 \label{eq:W-seq-k} W(\gamma) = \frac{q^{F(\seq{k})}}
  {\prod_{t=1}^{T} (q)_{k_{t}}}. 
\end{equation}
Note that there is a positive integer $\Delta$ such that $\Delta
F(\seq{k})\in \Z$ for all $\seq{k}\in \Z^{T}$. The mutation loop
$\gamma$ is called \emph{positive}, if $F(\seq{k})>0$ for all
$\seq{k}\in \N^{T}$, $\seq{k}\neq 0$, where
$\N=\{0,1,2,\cdots\}$. This condition assures finiteness of the set
$\{\seq{k}\in \N^{T}~|~ F(\seq{k})=n\}$ for all $n\in \Z$.

Now we are ready to define $Z(\gamma)$. From now on, mutation loops are
assumed to be nondegenerate and positive. Let $\gamma$ be a mutation
loop with $T$ mutations. We define its \emph{partition $q$-series}
$Z(\gamma)$ by the multiple sum
\begin{equation}
 \label{eq:Z-def} Z(\gamma)= \sum_{k_{1},\dots,k_{T}=0}^{\infty}
  W(\gamma) \quad \in \N[[q^{1/\Delta}]].
\end{equation}
%$Z(Q,\gamma)$ is simply written as $Z(\gamma)$ if the initial quiver
%$Q$ is clear from the context.

\begin{rem}
 \label{rem:k-vee}
 Occasionally it is convenient to introduce another set of variables,
 $k^{\vee}$-variables $\{k^{\vee}_{t}\}$ (see e.g. the proof of Theorem
 \ref{thm:Z(Q,Q')-multisum} below).  These are ``orientation reversed''
 version of $k$-variables, and the linear relation now reads
 \begin{equation}
 \label{eq:kv-and-s}
 k^{\vee}_{t} = s_{v}+s'_{v} -\sum_{v\to b}s_{b}.
 \end{equation}
Then, the weight of mutation \eqref{eq:W-def} is expressed as
\begin{equation}
 \label{eq:W-def-2}
  W(m_{t}) = \frac{q^{\frac{1}{2}k_{t}k^{\vee}_{t}}}{(q)_{k_{t}}}.
\end{equation}
\end{rem}

\subsection{Example 1 --- $A_{3}$ quiver}
\label{sec:A3}

We illustrate how to compute the partition $q$-series 
using the quiver of type $A_{3}$
\begin{equation*}
 Q= \xymatrix @R=6mm @C=6mm @M=4pt{ 1 \ar[r] & 2 & 3 \ar[l] } \strut
\end{equation*}
and a mutation loop 
\begin{equation}
 \gamma=(Q;\seq{m},\varphi),\qquad \seq{m}=(2,1,3),\qquad \varphi=\id.
\end{equation}
We label the $s$- and $k$-variables as follows:\footnote{Here the
$k$-variables are indexed by vertex labels instead of mutation order;
this is in accordance with the convention of Section \ref{sec:ZI}.}
\begin{equation}
 \label{eq:A3-loop}
 \xymatrix@R=6mm @C=6mm @M=4pt{ Q(0) \ar[d]_{\textstyle \mu_{2}} & s_1
  \ar[r] \ar@{=}[d] & s_2 \ar@{.>}[d]^{\textstyle k_{2}}& s_3
  \ar[l]\ar@{=}[d] 
\\ 
Q(1) \ar[d]_{\textstyle \mu_{1}}&
  s_1\ar@{.>}[d]^{\textstyle k_{1}} & s'_2 \ar[l]\ar[r]\ar@{=}[d]&
  s_3\ar@{=}[d] 
\\ 
Q(2) \ar[d]_{\textstyle \mu_{3}}& s'_1\ar[r]\ar@{=}[d]
  & s'_2 \ar[r]\ar@{=}[d]& s_3\ar@{.>}[d]^{\textstyle k_{3}} 
\\ 
Q(3) \ar[d]_{\textstyle \id}
&  s'_1\ar[r] \ar@{=}[d]& s'_2 \ar@{=}[d]& s'_3 \ar[l] \ar@{=}[d]
\\ 
Q(0) &  s_1\ar[r] & s_2 & s_3 \ar[l] }
\end{equation}
The relations between $k$- and $s$-variables are
\begin{equation}
 \label{eq:A3-ks}
 k_{2}=s_{2}+s'_{2}-s_{1}-s_{3}, \quad 
 k_{1}=s_{1}+s'_{1}-s'_{2},\quad  
 k_{3}=s_{3}+s'_{3}-s'_{2}.
\end{equation}
Under the boundary conditions $s_i=s'_i$ ($i=1,2,3$), one can solve
\eqref{eq:A3-ks} for $s$-variables:
\begin{equation*}
\begin{split}
 & s_1=s'_{1} =\frac{1}{4} \left(3 k_1+2 k_2+k_3\right),\qquad 
 s_2=s'_{2}= \frac{1}{2} \left(k_1+2 k_2+k_3\right),
\\
& s_3=s'_{3}= \frac{1}{4} \left(k_1+2 k_2+3 k_3\right).
\end{split}
\end{equation*}
So the weight of $\gamma$ takes the following form:
\begin{equation}
 \label{eq:A3-W}
  \begin{split}
   W(\gamma)&=
   \frac{q^{\frac{1}{2}(s_2+s'_{2}-s_1-s_3)(s_2+s'_{2})}}
   {(q)_{s_2+s'_{2}-s_1-s_3}} 
   \frac{q^{\frac{1}{2}(s_1+s'_{1}-s'_2)(s_1+s'_{1})}}
   {(q)_{s_1+s'_{1}-s'_2}} 
   \frac{q^{\frac{1}{2}(s_3+s'_{3}-s'_2)(s_3+s'_{3})}}
   {(q)_{s_3+s'_{3}-s'_2}}
   \\ &
   =\frac{q^{\frac{3}{4}k_1^2+k_1 k_2 +k_2^2 +k_2 k_3
   +\frac{3}{4} k_3^2 +\frac{1}{2}k_3 k_1
  }}{(q)_{k_1} (q)_{k_2} (q)_{k_3}}.
  \end{split}
\end{equation}
Summing over $k$-variables,
we obtain
\begin{equation}
  Z(\gamma)
  =\sum_{k_{1},k_{2},k_{3}=0}^{\infty} 
  \frac{q^{
  \frac{3}{4}k_1^2+k_1 k_2 +k_2^2 +k_2 k_3+\frac{3}{4} k_3^2 +\frac{1}{2}k_3 k_1
  }}{(q)_{k_1}
  (q)_{k_2} (q)_{k_3}}\in \N[[q^{1/4}]].
\end{equation}
Exactly the same formula appeared in the study of coset conformal field
theories \cite{Kedem1993}.  This is an example of partition $q$-series
which we study more systematically in Section \ref{sec:ZI}. $Z(\gamma)$
can be written as
\begin{equation}
  Z(\gamma) =\frac{1}{(q)_{\infty }}\sum_{n\in \Z} q^{\frac{3}{4}n^{2}},
\end{equation}
which reveals that $q^{-\frac{1}{24}}Z(\gamma)$ is a modular function for
a certain congruence subgroup of $SL_{2}(\Z)$.

\subsection{Example 2 --- pentagon identity}
\label{sec:A2-pentagon}

Let us consider a quiver of type $A_{2}$:
\begin{equation*}
 Q=1\longrightarrow 2~.
\end{equation*}
We take up two mutation loops $\gamma$, $\gamma'$, and compare the
associated partition $q$-series.

The first loop we study is 
\begin{equation}
 \label{eq:A2-g1}
 \gamma=(Q;\seq{m},\varphi),\qquad \seq{m}=(1,2),\qquad \varphi=\id.
\end{equation}
The $s$- and $k$-variables are given as follows:
\begin{equation}
\label{eq:A2-loop-1}
 \vcenter{\xymatrix @R=4mm @C=6mm @M=4pt{
  Q(0) \ar[d]_{\textstyle\mu_{1}} & 
   a \ar[r] \ar@{.>}[d]_{\textstyle  k_{1}}
   & b \ar@{=}[d]\\ 
 Q(1) \ar[d]_{\textstyle \mu_{2}} & 
   a' \ar@{=}[d]& b \ar[l]
   \ar@{.>}[d]^{\textstyle k_{2}}\\ 
 Q(2) \ar[d]_{\textstyle \id} & a' \ar[r] \ar@{=}[d]& b' \ar@{=}[d]\\
 Q(0) & a \ar[r] & b \\
}}
\end{equation}
Since there is no incoming arrow on mutating vertices, the relations
\eqref{eq:k-and-s} among $k$- and $s$-variables are simply
\begin{equation}
 \label{eq:A2-ks}
 k_{1}=a+a',\qquad k_{2}=b+b'.
\end{equation}
The initial and the new $s$-variables are identified via
boundary condition $\varphi=\id$:
\begin{equation}
 \label{eq:A2-s-id}
 a=a',\qquad b=b'.
\end{equation}
So there are two $s$-variables $a$, $b$ and two $k$-variables $k_{1}$,
$k_{2}$. Solving \eqref{eq:A2-ks} and \eqref{eq:A2-s-id} for
$s$-variables, we have
\begin{equation}
 a=a'=\frac{1}{2}k_{1},\qquad b=b'=\frac{1}{2}k_{2}.
\end{equation}
The weight of mutation loop is thus
\begin{equation}
 W(\gamma)=\frac{q^{\frac{1}{2}(a+a')(a+a'-b)}}{(q)_{a+a'}}
  \frac{q^{\frac{1}{2}(b+b')(b+b'-a')}}{(q)_{b+b'}} =
  \frac{q^{\frac{1}{2}(k_{1}^{2}-k_{1}k_{2}+k_{2}^{2})}}
  {(q)_{k_{1}}(q)_{k_{2}}}.
\end{equation}
The mutation loop $\gamma$ is nondegenerate and positive, because the quadratic
form $k_{1}^{2}-k_{1}k_{2}+k_{2}^{2}$ is positive definite. The partition
$q$-series is, by definition,
\begin{equation}
 \label{eq:pent-Z-1}
 Z(\gamma)=\sum_{k_{1},k_{2}=0}^{\infty} 
  \frac{q^{\frac{1}{2}(k_{1}^{2}-k_{1}k_{2}+k_{2}^{2})}}
  {(q)_{k_{1}}(q)_{k_{2}}} \in \N[[q^{1/2}]].
\end{equation}

The second loop we consider is 
\begin{equation}
 \label{eq:A2-g2}
 \gamma'=(Q;\seq{m}',\varphi'),\qquad \seq{m}'=(2,1,2),\qquad \varphi'=(12),
\end{equation}
where $\varphi'=(12)$ means the transposition of the two vertices.  The
$s$- and $k$-variables are given as follows:
\begin{equation}
 \label{eq:A2-loop-2}
 \vcenter{\xymatrix @R=5mm @C=6mm @M=4pt{
 Q(0) \ar[d]_{\textstyle\mu_{2}} & a \ar[r] \ar@{=}[d]& b
  \ar@{.>}[d]^{\textstyle k'_{1}}\\
 Q(1) \ar[d]_{\textstyle\mu_{1}} &  
  a \ar@{.>}[d]^{\textstyle k'_{2}}& b' \ar[l]  \ar@{=}[d]\\
 Q(2) \ar[d]_{\textstyle\mu_{2}} & a' \ar[r] \ar@{=}[d]& 
  b' \ar@{.>}[d]^{\textstyle k'_{3}}\\
 Q(3) \ar[d]_{\textstyle (12)}
  &  a' \ar@{=}[dr]& b'' \ar[l] \ar@{=}[dl]\\
 Q(0) & a \ar[r] & b
   }}~.
\end{equation}
The relations between $k$- and $s$-variables are 
\begin{equation}
 \label{eq:A2-ks-2}
 k'_{1}=b+b'-a,\qquad k'_{2}=a+a'-b',\qquad k'_{3}=b'+b''-a'.
\end{equation}
The boundary condition implies $a'=b$ and $b''=a$. Taking this into the
account, we can solve \eqref{eq:A2-ks-2} for $s$-variables:
\begin{equation}
 a=b''=\frac{k'_2}{2}+\frac{k'_3}{2},
  \qquad b=a'= \frac{k'_1}{2}+\frac{k'_2}{2},
  \qquad  b'= \frac{k'_1}{2}+\frac{k'_3}{2}.
\end{equation}
The weight for the mutation loop $\gamma'$ is thus
\begin{equation}
\begin{split}
 W(\gamma') & = \frac{q^{\frac{1}{2}(b+b'-a)(b+b')}}{(q)_{b+b'-a}}
 \frac{q^{\frac{1}{2} (a+a'-b')(a+a')}}{(q)_{a+a'-b'}}
 \frac{q^{\frac{1}{2}  (b'+b''-a')(b'+b'')}}{(q)_{b'+b''-a'}}
 \\
 &= 
 \frac{q^{\frac{1}{2} \left( k'_1{}^2 +k'_2{}^2+k'_3{}^2 +k'_1 k'_2 
 +k'_2 k'_3 +k'_3 k'_1\right)}}{(q)_{k'_{1}}(q)_{k'_{2}}(q)_{k'_{3}}}.
\end{split}
\end{equation}
The partition $q$-series is now defined as
\begin{equation}
 \label{eq:pent-Z-2}
 Z(\gamma')=\sum_{k'_{1},k'_{2},k'_{3}=0}^{\infty}
 \frac{q^{\frac{1}{2} \left(
 k'_1{}^2
 +k'_2{}^2+k'_3{}^2
 +k'_1 k'_2 
 +k'_2 k'_3
 +k'_3 k'_1
\right)}}{(q)_{k'_{1}}(q)_{k'_{2}}(q)_{k'_{3}}}.
\end{equation}
It turns out that the partition $q$-series \eqref{eq:pent-Z-1} and \eqref{eq:pent-Z-2} 
are equal due to the identity (see e.g. \cite{Zagier2007})
\begin{equation}
 \label{eq:q-pentagon}
 \frac{1}{(q)_{m}(q)_{n}}=
  \sum_{\substack{r,s,t\geq 0\\m=r+s\\ n=s+t}}
  \frac{q^{rt}}{(q)_{r}(q)_{s}(q)_{t}}.
\end{equation}
This is no coincidence. In the next section, we state and prove a
general result about the conditions on mutation loops, which guarantee
the equality of associated partition $q$-series.

\section{Generalized pentagon identity}
\label{sec:Invariance}

The main result of this section is Theorem
\ref{thm:pentagon-move-invariance}, saying that as a function of
mutation loops, the partition $q$-series $Z(\gamma)$ is
invariant under \emph{pentagon move} of $\gamma$, which we define shortly.

\subsection{Pentagon move}

It is convenient to slightly generalize the notion of mutation
sequences/loops to keep track of vertex relabeling effect.  Let $Q$ be a
quiver with vertices $\{1,\cdots,n\}$. A finite sequence
$\seq{f}=(f_{1},\dots,f_{r})$ consisting of
\begin{itemize}
 \item[(a)] mutation $\mu_{i}$ at the vertex $i$ ($1\leq i\leq n$), or
 \item[(b)] vertex relabeling by an element
      $\sigma$ of the symmetric group $S_{n}$,
\end{itemize}
is called a \emph{mutation sequence}. If $\seq{f}(Q):=f_{r}(\cdots
(f_{1}(Q))\cdots )$ is isomorphic to $Q$ as a (labeled) quiver, then
$(Q;\seq{f})$ is called a \emph{mutation loop}. Two
sequences $\seq{f}$, $\seq{f}'$ are considered equivalent if they are
related by a series of the following moves (rewriting rules):
\begin{itemize}
 \item $(\cdots,\sigma_{1},\sigma_{2},\cdots ) \simeq
   (\cdots,\sigma_{2}\circ \sigma_{1},\cdots )$,
 \item $(\cdots,\mu_{i},\sigma,\cdots ) \simeq
  (\cdots,\sigma,\mu_{\sigma(i)},\cdots )$,
 \item $(\cdots,\id,\cdots ) \simeq
  (\cdots,\cdots )$.
\end{itemize}
Clearly any mutation sequence is equivalent to the form of 
$(\mu_{m_{1}},\dots,\mu_{m_{T}},\varphi)$, $\varphi\in S_{n}$, i.e.
a pair of a mutation sequence and a boundary condition.

\begin{figure}[bth]
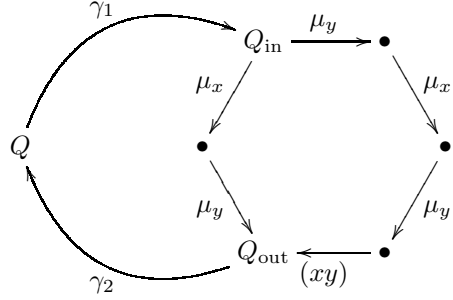

 \begin{equation*}
  \vcenter{ \xygraph{!{<0cm,0cm>;<23pt,0cm>:<0cm,20pt>::} !{(0,2)
  }*+{Q}="Q0" !{(4,4) }*+{Q_{\text{in}}}="Q1" !{(4,0)
  }*+{Q_{\text{out}}}="Q2" !{(3,2) }*+{\bullet}="Q3" !{(6,4)
  }*+{\bullet}="Q4" !{(7,2) }*+{\bullet}="Q5" !{(6,0)
  }*+{\bullet}="Q6"
  "Q0":@/^1cm/"Q1"^{\textstyle {\gamma_{1}}}
  "Q2":@/^1cm/"Q0"^{\textstyle {\gamma_{2}}} "Q1":"Q3"_{\textstyle
  \mu_{x}} "Q3":"Q2"_{\textstyle \mu_{y}} "Q1":"Q4"^{\textstyle \mu_{y}}
  "Q4":"Q5"^{\textstyle \mu_{x}} "Q5":"Q6"^{\textstyle \mu_{y}}
  "Q6":"Q2"^{\textstyle (xy)} }}
 \end{equation*}
 \caption{Pentagon move. $\gamma=(Q;\gamma_{1},\mu_{x},\mu_{y},
 \gamma_{2})$ and $
 \gamma'=(Q;\gamma_{1},\mu_{y},\mu_{x},\mu_{y},(xy),\gamma_{2})$.}\label{fig:pentagon-move}
\end{figure}

A local change of mutation loops of the following type
\begin{equation}
 \label{eq:pentagon-equiv}
  \gamma=(Q;\gamma_{1},\mu_{x},\mu_{y}, \gamma_{2})
  \qquad \longleftrightarrow \qquad
  \gamma'=(Q;\gamma_{1},\mu_{y},\mu_{x},\mu_{y},(xy),\gamma_{2})
\end{equation}
is called \emph{pentagon move}.  Here
$\gamma_{1}$, $\gamma_{2}$ are arbitrary mutation sub-sequences, and the
vertices $x,y$ are assumed to be connected by a single arrow $x\to y$ in
$Q_{\text{in}}:=\gamma_{1}(Q)$.
This condition guarantees that $(\mu_{x},\mu_{y})(Q_{\text{in}})$ and
$(\mu_{y},\mu_{x},\mu_{y},(xy))(Q_{\text{in}})$ are isomorphic as
labeled quivers: we denote this quiver by $Q_{\text{out}}$.  (Figure \ref{fig:pentagon-move})

The mutation loops of Section \ref{sec:A2-pentagon} are the simplest
example of those related by a pentagon move.  For another example, take
$ Q=(1\rightarrow 2 \leftarrow 3 \rightarrow 4)$.  The mutation loop
\begin{equation*}
 \gamma=(Q;\seq{m},\varphi),\qquad 
\seq{m}=(4,1,2,3,2,4,1), \qquad \varphi= \left(\begin{smallmatrix} 1 & 2 & 3
& 4 \\ 4 & 1 & 2 & 3\end{smallmatrix}\right)
\end{equation*}
is, via pentagon move, equivalent to 
\begin{equation*}
 \gamma'=(Q;\seq{m}',\varphi'),\qquad \seq{m}=(4,2,1,2,3,1,4,2),\qquad 
\varphi'= \left(\begin{smallmatrix} 1 & 2 & 3 & 4 \\ 1 & 4& 2 &
3\end{smallmatrix}\right). 
\end{equation*}
Indeed,
\begin{equation*}
\begin{split}
 \gamma=& (Q;\mu_{4},\underline{\strut
 \mu_{1},\mu_{2}},\mu_{3},\mu_{2},\mu_{4}, \mu_{1},\varphi)
 \\
 \rightarrow &
 (Q;\mu_{4},\underline{\mu_{2},\mu_{1},\mu_{2},(12)}, \mu_{3},\mu_{2},\mu_{4}, \mu_{1}, \varphi)
 \\
 \simeq &
 (Q;\mu_{4},\mu_{2},\mu_{1},\mu_{2}, \mu_{3},\mu_{1},\mu_{4},\mu_{2}, (12), \varphi)
 \\
 \simeq &
 (Q;\mu_{4},\mu_{2},\mu_{1},\mu_{2}, \mu_{3},\mu_{1},\mu_{4},\mu_{2},
  \varphi')=\gamma'.
\end{split}
\end{equation*}

\subsection{Generalized pentagon identity}
\label{sec:gen-pent}

The main result of this section is the next
\begin{theorem}\label{thm:pentagon-move-invariance} 
 The partition $q$-series $Z(\gamma)$ is invariant under the pentagon
 move of the loop $\gamma$; that is, for the mutation loops $\gamma$,
 $\gamma'$ in \eqref{eq:pentagon-equiv}, we have
 \begin{equation*}
  Z(\gamma)=Z(\gamma').
 \end{equation*}
\end{theorem}

The rest of this subsection is devoted to the proof of Theorem
\ref{thm:pentagon-move-invariance}.  The key idea is to cut the mutation
loops $\gamma$, $\gamma'$ at $Q_{\text{in}}$ and $Q_{\text{out}}$ into
two pieces --- ``internal part'' and ``external part'' (Figure
\ref{fig:pentagon-move}), and treat their contribution to the mutation
weights separately. We put
\begin{equation}
 \begin{aligned}
 \gamma &=(Q;\gamma_{1} , \seq{m}, \gamma_{2} )
  &&& \seq{m}&=(\mu_{x},\mu_{y}),
  \\
 \gamma' &=(Q;\gamma_{1} , \seq{m}' , \gamma_{2} )
  &&& \seq{m}'&=(\mu_{y},\mu_{x},\mu_{y},(xy)).
 \end{aligned}
\end{equation}
The subsequences $\seq{m}$, $\seq{m}'$ from $Q_{\text{in}}$ to
$Q_{\text{out}}$ is referred to as ``internal''; the rest is considered
as ``external''.

By definition of pentagon move, two mutating vertices $x$, $y$ are
connected by a single arrow in $Q_{\text{in}}$ (Figure
\ref{fig:near-xy}). The vertices $x$ and $y$ can be a source or a target
of other arrows in $Q_{\text{in}}$; such arrows are collectively denoted
as $a_{i}\to x$, $b_{j}\to y$, $x\to c_{k}$ and $y\to d_{l}$.  Along the
mutation paths $\seq{m}$, $\seq{m}'$ from $Q_{\text{in}}$ to
$Q_{\text{out}}$, the quiver will change as in Figure
\ref{fig:general-pentagon}.

\begin{figure}[th]
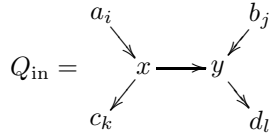

\begin{equation*}
 Q_{\text{in}}=
 \vcenter{
  \xygraph{!{<0cm,0cm>;<20pt,0cm>:<0cm,20pt>::}
  !{(0.8,1) }*+<4pt>{x}="x"
  !{(2.2,1) }*+<4pt>{y}="y"
  !{(0,2) }*+<3pt>{a_{i}}="a"
  !{(3,2) }*+<1pt>{b_{j}}="b"
  !{(0,0) }*+<3pt>{c_{k}}="c"
  !{(3,0) }*+<2pt>{d_{l}}="d"
  "x":"y"
  "a":"x"
  "x":"c"
  "b":"y"
  "y":"d"
  }}
\end{equation*}
 \caption{The quiver $Q_{\text{in}}$. Only the arrows incident on $x$
 or $y$ are shown. Some of the vertices $a_{i}$, $b_{j}$, $c_{k}$,
 $d_{l}$ may be missing, duplicated or identified.}\label{fig:near-xy}
\end{figure}

\begin{figure}[th]
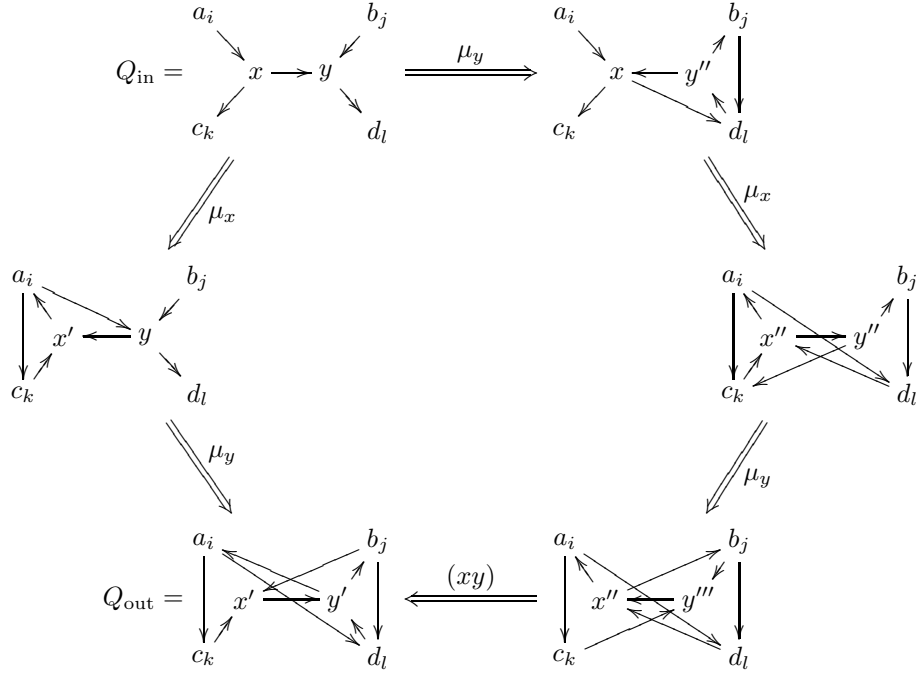

\begin{equation*}
 \xygraph{!{<0cm,0cm>;<16mm,0cm>:<0cm,35mm>::}
!{(1,3)}*+{\llap{$Q_{\text{in}}=$}
\vcenter{\xybox{\xygraph{!{<0cm,0cm>;<22pt,0cm>:<0cm,22pt>::}
!{(0.9,1) }*+{x}="x" !{(2.1,1) }*+{y}="y" !{(0,2) }*+{a_{i}}="a" !{(0,0)
}*+{c_{k}}="c" !{(3,2) }*+{b_{j}}="b" !{(3,0) }*+{d_{l}}="d" "x":"y"
"a":"x" "x":"c" "b":"y" "y":"d" }} }}="Q1"
!{(-0.5,2)}*+{\xybox{\xygraph{!{<0cm,0cm>;<22pt,0cm>:<0cm,22pt>::}
!{(0.7,1) }*+{x'}="x" !{(2.1,1) }*+{y}="y" !{(0,2) }*+{a_{i}}="a"
!{(0,0) }*+{c_{k}}="c" !{(3,2) }*+{b_{j}}="b" !{(3,0) }*+{d_{l}}="d"
"x":"a" "a":"c" "a":"y" "c":"x" "b":"y" "y":"d" "y":"x" }}}="Q2"
!{(1,1)}*+{\llap{$Q_{\text{out}}=$}
\vcenter{\xybox{\xygraph{!{<0cm,0cm>;<22pt,0cm>:<0cm,22pt>::}
!{(0.7,1) }*+{x'}="x" !{(2.3,1) }*+{y'}="y" !{(0,2) }*+{a_{i}}="a"
!{(0,0) }*+{c_{k}}="c" !{(3,2) }*+{b_{j}}="b" !{(3,0) }*+{d_{l}}="d"
"x":"y" "a":"c" "a":"d" "c":"x" "b":"x" "b":"d" "d":"y" "y":"a" "y":"b"
}}} }="Q3"
!{(4,3)}*+{\xybox{\xygraph{!{<0cm,0cm>;<22pt,0cm>:<0cm,22pt>::}
!{(0.9,1) }*+{x}="x" !{(2.3,1) }*+{y''}="y" !{(0,2) }*+{a_{i}}="a"
!{(0,0) }*+{c_{k}}="c" !{(3,2) }*+{b_{j}}="b" !{(3,0) }*+{d_{l}}="d"
"y":"x" "a":"x" "x":"c" "y":"b" "d":"y" "x":"d" "b":"d" }}}="Q4"
!{(5.4,2)}*+{\xybox{\xygraph{!{<0cm,0cm>;<22pt,0cm>:<0cm,22pt>::}
!{(0.7,1) }*+{x''}="x" !{(2.3,1) }*+{y''}="y" !{(0,2) }*+{a_{i}}="a"
!{(0,0) }*+{c_{k}}="c" !{(3,2) }*+{b_{j}}="b" !{(3,0) }*+{d_{l}}="d"
"x":"y" "x":"a" "a":"c" "a":"d" "c":"x" "b":"d" "d":"x" "y":"c" "y":"b"
}}}="Q5" 
!{(4,1)}*+{\xybox{\xygraph{!{<0cm,0cm>;<22pt,0cm>:<0cm,22pt>::}
!{(0.7,1) }*+{x''}="x" !{(2.3,1) }*+{y'''}="y" !{(0,2) }*+{a_{i}}="a"
!{(0,0) }*+{c_{k}}="c" !{(3,2) }*+{b_{j}}="b" !{(3,0) }*+{d_{l}}="d"
"y":"x" "x":"a" "x":"b" "a":"c" "a":"d" "c":"y" "b":"y" "b":"d" "d":"x"
}}}="Q6" "Q1":@{=>}^{\textstyle \mu_{x}}"Q2" "Q2":@{=>}^{\textstyle
\mu_{y}}"Q3" "Q1":@{=>}^{\textstyle \mu_{y}}"Q4" "Q4":@{=>}^{\textstyle
\mu_{x}}"Q5" "Q5":@{=>}^{\textstyle \mu_{y}}"Q6" "Q6":@{=>}_{\textstyle
(xy)}"Q3" }
\end{equation*}
 \caption{Pentagon move and quiver mutations. The vertices of quivers
 are represented by the corresponding $s$-variables.  $\seq{m}=(\mu_x,\mu_y)$ and
 $\seq{m}'=(\mu_y,\mu_x,\mu_y,(xy))$ represent two internal paths from
 $Q_{\text{in}}$ to $Q_{\text{out}}$. We have $x''=y'$, $y'''=x'$ via
 transposition $(xy)$ of vertices.}\label{fig:general-pentagon}
\end{figure}

Let $k_{1},k_{2}$ be the $k$-variables associated with
$\seq{m}=(\mu_x,\mu_y)$, and $k_{3},k_{4},k_{5}$ be those for
$\seq{m}'=(\mu_y,\mu_x,\mu_y, (xy))$.  The $k$-variables on external
part are denoted by $\{ l_i \}$; they are common to both $\gamma$ and
$\gamma'$. For the internal part, the $k$- and $s$-variables are related as
\begin{equation}
 \label{eq:GP-ks}
 \begin{aligned}
  k_{1}&=x+x'-\sum a_{i}, \\
  k_{2}&=y+y'-\sum a_{i} -\sum b_{j},  \\
  k_{3}& =y+y''-\sum b_j-x, \\
  k_{4}& = x+x''-\sum a_i-y'',\\
  k_{5}&=y''+y'''-x''. 
\end{aligned}
\end{equation}
The relations \eqref{eq:GP-ks} and the identification $y'=x''$, $x'=y'''$
yield the constraint
\begin{equation}
  \label{eq:k-sum-const}
   k_{3}+k_{4}= k_{2},\qquad  k_{4}+k_{5}=k_{1}.
\end{equation}
The partition $q$-series has the following form:
\begin{equation}
 \label{eq:GP-Zs}
  \begin{aligned}
 Z(\gamma )&=\sum_{k_1,k_2,l_i \geq 0}
 \frac{q^{F_{\text{ext}}(k_1,k_2,l)}}{\prod_i (q)_{l_i}}\times
 \frac{q^{F_{\text{int}}(k_1,k_2,l)}}{(q)_{k_1}(q)_{k_2}}
   \\ &=\sum_{l_i \geq 0} \frac{1}{\prod_i (q)_{l_i}} 
   \times \sum_{k_1,k_2\geq 0} 
   \frac{q^{F_{\text{ext}}(k_1,k_2,l)+
   F_{\text{int}}(k_1,k_2,l)}}{(q)_{k_1}(q)_{k_2}},
   \\[5pt]
   Z(\gamma' )&=\sum_{k_3,k_4,k_5, l_i \geq 0}
   \frac{q^{F'_{\text{ext}}(k_3,k_4,k_5,l)}}{\prod_i (q)_{l_i}} \times
   \frac{q^{F'_{\text{int}}(k_3,k_4,k_5,l)}}{(q)_{k_3}(q)_{k_4}(q)_{k_5}}
   \\ &=\sum_{l_i \geq 0} \frac{1}{\prod_i (q)_{l_i}} \times
   \sum_{k_3,k_4,k_5 \geq 0}
   \frac{q^{F'_{\text{ext}}(k_3,k_4,k_5,l)+F'_{\text{int}}(k_3,k_4,k_5,l)}}
   {(q)_{k_3}(q)_{k_4}(q)_{k_5}}.
  \end{aligned}
\end{equation}
The external part of $\gamma$ and $\gamma'$ share the same set of
$s$-variables, so as functions $s$-variables,
$F_{\text{ext}}=F'_{\text{ext}}$. Therefore under the identification
\eqref{eq:k-sum-const}, we have
\begin{equation}
 \label{eq:Fext-rel}
  \left(
 F_{\text{ext}}(k_1,k_2,l)
 \biggm|\!{}_{\substack{k_{1}=k_{4}+k_{5}\\k_{2}=k_{3}+k_{4}}}
\right)
 =F'_{\text{ext}}(k_3,k_4,k_5,l).
\end{equation}
As for the internal part, we obtain after some computation,
\begin{equation*}
 \begin{split}
  F_{\text{int}}(k_1,k_2,l) 
  & = \frac{1}{2}
  (x+x'-\sum a_i)(x+x'-\sum c_k-y)
    \\  
  & \qquad 
   +  \frac{1}{2}(y+y'-\sum a_i-\sum b_j)(y+y'-x'-\sum d_l)
  \\
  & = \frac{1}{2}\left(
  k_1^2+k_2^2-k_1k_2+Ak_{1}+Bk_{2} \right), 
  \\[8pt]
  F'_{\text{int}}(k_3,k_4,k_5,l) 
  & =  \frac{1}{2} (y+y''-\sum b_j-x)(y+y''-\sum d_l)\\
  & \qquad   +\frac{1}{2}(x+x''-\sum a_i-y'')(x+x''-\sum c_k-\sum d_l) \\
  & \qquad   +\frac{1}{2}(y''+y'''-x'')(y''+y'''-\sum b_j-\sum c_k) 
  \\
  &= \frac{1}{2} \bigl(
  k_3^2+k_4^2+k_5^2+k_3k_4+k_4k_5+k_3k_5 +Bk_{3} +(A+B)k_{4}+Ak_{5}
  \bigr),
 \end{split}
\end{equation*}
where
\begin{equation*}
 A:=\sum a_i-\sum c_k-y, \qquad B :=  \sum b_j-\sum d_k+x.
\end{equation*}
It is now easy to check that under the relation \eqref{eq:k-sum-const},
\begin{equation}
 \label{eq:Fint-rel}
 F'_{\text{int}}(k_3,k_4,k_5,l) = 
\left(
F_{\text{int}}(k_1,k_2,l)\biggm|\!{}_{\substack{k_{1}=k_{4}+k_{5}\\k_{2}=k_{3}+k_{4}}}
\right) + k_{3}k_{5}.
\end{equation}
By substituting \eqref{eq:Fext-rel} and \eqref{eq:Fint-rel} into
\eqref{eq:GP-Zs}, we conclude $Z(\gamma )$ and $Z(\gamma')$ are equal
thanks to \eqref{eq:q-pentagon}. This completes the proof of Theorem 
\ref{thm:pentagon-move-invariance}.

\section{Partition $q$-series and Characters I --- Dynkin case}
\label{sec:ZI}

Let $Q$ be an alternating quiver of Dynkin type $A_{n}$, $D_{n}$ or
$E_{n}$. Denote by $\seq{m}_{+}$, $\seq{m}_{-}$ the set of sources,
sinks of $Q$, respectively.  We consider the following mutation sequence
of length $n=\#Q_{0}$:
\begin{equation}
 \label{eq:mu=mu-mu+}
 \seq{m}= \seq{m}_{-}\seq{m}_{+}.
\end{equation}
Here consider $\seq{m}_{\pm}$ as sequence of mutations. The ordering
within $\seq{m}_{\pm}$ does not matter since there are no arrows
connecting two sources or two sinks in $Q$.  It is easy to check that
\begin{equation*}
 \mu_{\seq{m}_{-}}(Q) = Q^{op},\qquad
  \mu_{\seq{m}_{+}}(Q^{op})=Q.
\end{equation*}
Thus with trivial boundary condition $\varphi=\id$, $\gamma=
(Q;\seq{m},\id)$ makes up a mutation loop. The $s$-variables $s_{v}$ and
$s'_{v}$, before and after the mutation at $v$, are identified for each
$v\in Q_{0}$. Since every vertex $v$ of $Q$ is mutated exactly once, it
is convenient to label the $k$-variables by vertices, not by mutation
order; we use the notation $\seq{k}=(k_{v})_{v\in Q_{0}}$.

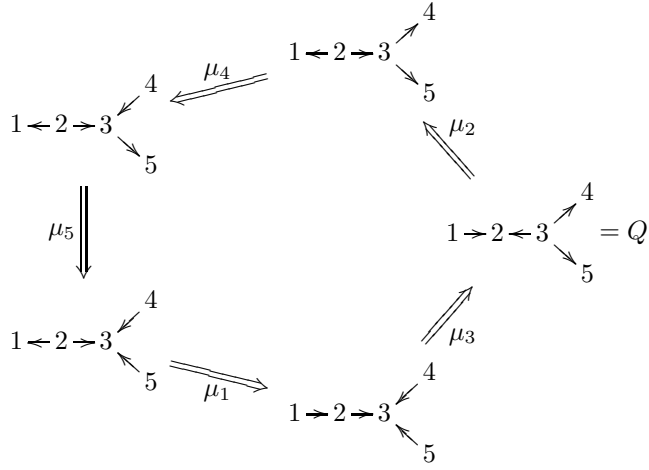
\begin{figure}[bhtp]
\begin{equation*}
 \xygraph{!{<0cm,0cm>;<10mm,0cm>:<0cm,6mm>::}
  !{(5.8,0)}*{
  \vcenter{\xybox{\xymatrix @R=5pt @C=8pt @M=2pt{ & & & {4} \\
   {1} \ar[r]  &{2} & {3} \ar[l]  \ar[ur] \ar[dr]  \\
   &&& {5}  }}}\rlap{$=Q$}
    }="Q0"
  !{(3.7,4)}*++{
   \vcenter{\xybox{\xymatrix @R=5pt @C=8pt @M=2pt{ & & & {4} \\
   {1}    &{2} \ar[l]\ar[r] & {3}  \ar[ur] \ar[dr]  \\
   &&& {5}  }}}}="Q1"
  !{(0,2.4)}*+{
   \vcenter{\xybox{\xymatrix @R=5pt @C=8pt @M=2pt{ & & & {4}  \ar[dl] \\
   {1}    &{2} \ar[l]\ar[r] & {3}  \ar[dr]  \\
   &&& {5}  }}}}="Q2"
  !{(0,-2.4)}*+{
   \vcenter{\xybox{\xymatrix @R=5pt @C=8pt @M=2pt{ & & & {4}  \ar[dl] \\
   {1}    &{2} \ar[l]\ar[r] & {3} \\
   &&& {5}  \ar[ul]   }}}}="Q3"
  !{(3.7,-4)}*++{
   \vcenter{\xybox{\xymatrix @R=5pt @C=8pt @M=2pt{ & & & {4}  \ar[dl] \\
   {1} \ar[r] &{2} \ar[r] & {3} \\
   &&& {5}  \ar[ul]   }}}}="Q4"
"Q0":@{=>}_{\textstyle \mu_{2}}"Q1"
"Q1":@{=>}_{\textstyle \mu_{4}}"Q2"
"Q2":@{=>}_{\textstyle \mu_{5}}"Q3"
"Q3":@{=>}_{\textstyle \mu_{1}}"Q4"
"Q4":@{=>}_{\textstyle \mu_{3}}"Q0"}
\end{equation*}
\caption{Example of type $D_{5}$: the mutation loop
$\gamma=(Q; (2, 4, 5, 1, 3),\id)$.}\label{fig:D5-example}
\end{figure}

To motivate our main result of this section, Theorem
\ref{thm:Z-multisum}, we first give an example.  Consider an alternating
quiver $Q$ of type $D_{5}$
\begin{equation*}
 Q~=~ \vcenter{\xymatrix @R=3mm @C=5mm @M=2pt{ & & & {4} \\
 {1} \ar[r]  &{2} & {3} \ar[l]  \ar[ur] \ar[dr]  \\
 &&& {5}  }}
\end{equation*}
and the mutation sequence 
\begin{equation*}
 \seq{m}_{-}=(2, 4, 5),\quad  \seq{m}_{+}=(1, 3),\quad 
 \gamma=(Q;\seq{m}_{-}\seq{m}_{+},\id)=(Q;(2, 4, 5, 1, 3),\id).
\end{equation*}
See Figure \ref{fig:D5-example}.

The linear relations between $k$- and $s$-variables are
\begin{equation}
\label{eq:k2s}
\begin{split}
 & k_2= -s_1+2 s_2-s_3,\quad  k_4= 2 s_4-s_3,\quad k_5= 2 s_5-s_3,\\
 &  k_1= 2 s_1-s_2,\quad k_3= -s_2+2 s_3-s_4-s_5.
\end{split}
\end{equation}
Recall that $s'_{v}=s_{v}$ by the boundary condition. The weight is then
expressed as
\begin{equation}
 \label{eq:Z-D5-S}
 \begin{split}
 W(\gamma)    &=
  \frac{q^{\frac{1}{2} (2 s_2-s_1-s_3) 2s_2}}{(q)_{2 s_2-s_1-s_3}} 
  \frac{q^{\frac{1}{2} (2 s_4-s_3) 2s_{4}}}{(q)_{2 s_4-s_3}} 
  \frac{q^{\frac{1}{2} (2 s_5-s_3) 2s_{5}}}{(q)_{2 s_5-s_3}} 
  \frac{q^{\frac{1}{2} (2 s_1-s_2) 2s_{1}}}{(q)_{2 s_1-s_2}}
  \frac{q^{\frac{1}{2} (2s_3-s_2-s_4-s_5) 2s_{3}}}{(q)_{2 s_3-s_2-s_4-s_5}}
  \\
  & = \frac{q^{2 s_1^2-2 s_2 s_1+2 s_2^2+2
s_3^2+2 s_4^2+2 s_5^2-2 s_2 s_3-2 s_3 s_4-2 s_3 s_5}}{
(q)_{2 s_2-s_1-s_3}
(q)_{2 s_4-s_3}
(q)_{2 s_5-s_3}
(q)_{2 s_1-s_2}
(q)_{2 s_3-s_2-s_4-s_5}
}.
 \end{split}
\end{equation}
The relation \eqref{eq:k2s} is nondegenerate: we can solve
\eqref{eq:k2s} for $s$-variables:
\begin{equation}
 \label{eq:s2k}
\left\{
\begin{aligned}
 s_1 & =  \left(2 k_1+2 k_2+2 k_3+k_4+k_5\right)/2, \\
 s_2 & =  k_1+2 k_2+2 k_3+k_4+k_5, \\
 s_3 & =  \left(2 k_1+4 k_2+6 k_3+3 k_4+3 k_5\right)/2, \\
 s_4 & =  \left(2 k_1+4 k_2+6 k_3+5 k_4+3 k_5\right)/4, \\
 s_5 & =  \left(2 k_1+4 k_2+6 k_3+3 k_4+5 k_5\right)/4.
\end{aligned}
\right.
\end{equation}
Substituting these into \eqref{eq:Z-D5-S}, we can express $Z(\gamma)$ in
terms of $k$-variables alone:
\begin{equation}
 \label{eq:Z-D5-K} Z(\gamma)=\sum_{k_{1},\dots,k_{5}=0}^{\infty} 
  \frac{q^{k_1^2+2 k_2^2+3 k_3^2+\frac{5}{4}k_4^2+\frac{5}{4} k_5^2
  +2 k_1 k_2 +2 k_1 k_3 + k_1 k_4+ k_1 k_5 +4 k_2 k_3+2 k_2 k_4+2
  k_2 k_5+3 k_3 k_4+3 k_3 k_5+\frac{3}{2}k_4 k_5}}
  {(q)_{k_1} (q)_{k_2} (q)_{k_3} (q)_{k_4} (q)_{k_5}}.
\end{equation}

Let $A[\seq{x}]$ denote the quadratic form associated
with a symmetric $n\times n$ matrix $A=(a_{ij})$:
\begin{equation}
 A[\seq{x}]=\sum_{i,j=1}^{n} a_{ij}x_{i}x_{j} = \seq{x}^{T}A\seq{x}, \qquad
 (\seq{x}=(x_{1},\dots,x_{n})).
\end{equation}
The exponents of $q$ in the summand \eqref{eq:Z-D5-S} or
\eqref{eq:Z-D5-K} are quadratic form in $s$- or $k$-variables; they are
neatly expressed as $C[\seq{s}]$ and $D[\seq{k}]$, respectively, where
\begin{equation}
C= 
{\arraycolsep=2pt
\left(
\begin{array}{rrrrr}
 2 & -1 & 0 & 0 & 0 \\
 -1 & 2 & -1 & 0 & 0 \\
 0 & -1 & 2 & -1 & -1 \\
 0 & 0 & -1 & 2 & 0 \\
 0 & 0 & -1 & 0 & 2
\end{array}
\right)}, 
\qquad 
D=C^{-1}=\frac{1}{4}
{\arraycolsep=4pt
\left(
\begin{array}{ccccc}
 4 & 4 & 4 & 2 & 2 \\
 4 & 8 & 8 & 4 & 4 \\
 4 & 8 & 12 & 6 & 6 \\
 2 & 4 & 6 & 5 & 3 \\
 2 & 4 & 6 & 3 & 5
\end{array}
\right)}
\end{equation}
are nothing but the Cartan matrix of type $D_{5}$ and its inverse.  The
linear relations \eqref{eq:k2s} and \eqref{eq:s2k} are also simply given
by
\begin{equation}
 \seq{k}=C\seq{s},\qquad 
 \seq{s}=D\seq{k}.
\end{equation}

We write the product of $q$-Pochhammer symbols as
\begin{equation}
 (q)_{\seq{v}} :=  \prod_{i\in I} (q)_{v_{i}},
\end{equation}
where $\seq{v}=(v_{i})_{i\in I}$ is a vector of nonnegative integers.
The denominators of the weights are then simply expressed as
$(q)_{\seq{k}}$.

\begin{theorem}\label{thm:Z-multisum} Let $Q$ be an alternating
 quiver of simply-laced Dynkin type $X_{n}$ ($X_{n}=A_{n},D_{n},E_{n}$),
 and $ \gamma= (Q;\seq{m}_{-}\seq{m}_{+},\id)$ be the mutation loop
 defined in \eqref{eq:mu=mu-mu+}. Then the partition $q$-series
 $Z(\gamma)$ has a following form:
 \begin{equation}
  \label{eq:Z(m)-Dynkin}
  Z(\gamma)
   =\sum_{\seq{k}=(k_{1},\dots,k_{n})\in \N^{n}}
   \frac {q^{D[\seq{k}]}}{ (q)_{\seq{k}}}. 
 \end{equation}
 Here $D$ is the inverse of the Cartan matrix $C=(c_{ij})$ of type
 $X_{n}$.  The relation between $k$- and $s$-variables is nondegenerate
 and is given by $\seq{k}=C\seq{s}$.
\end{theorem}
\begin{proof}
 First consider the mutation sequence $\seq{m}_{-}$ applied on $Q$. It
 is important to note that every mutation vertex $a\in \seq{m}_{-}$ is a
 sink of $Q$.  So we have
 \begin{equation}
  \label{eq:ZI-ks-1}
  k_{a}=s_{a}+s'_{a}-\sum_{i\to a\in Q} s_{i} 
  =2s_{a}-\sum_{i\sim a\in \underline{Q}} s_{i},\qquad (a\in \seq{m}_{-}),
 \end{equation}
 where $i\sim a$ means that the vertices $i$ and $a$ are adjacent in the
 underlying Dynkin graph $\underline{Q}$.
 Here we used the identification $s_{a}'=s_{a}$. 
 The weight of the
 mutation at sink $a\in \seq{m}_{-}$ is
 \begin{equation}
  \label{eq:wt-}
  \frac{ q^{\frac{1}{2} (2s_{a}-\sum_{i\sim a} s_{i})\cdot(2s_{a}-0)}}
 {(q)_{k_{a}}} = \frac{q^{\sum_{i=1}^{n} c_{ia}
 s_{i}s_{a}}}{(q)_{k_{a}}}.
 \end{equation}

 Next consider the mutation sequence $\seq{m}_{+}$ on
 $\seq{m}_{-}(Q)=Q^{op}$. Again, every mutating vertex $b\in
 \seq{m}_{+}$ is a sink of $Q^{op}$. Therefore 
 \begin{equation}
  \label{eq:ZI-ks-2}
  k_{b}=s_{b}+s'_{b}-\sum_{i\to b\in Q^{op}} s_{i} 
  =2s_{b}-\sum_{i\sim b\in \underline{Q}} s_{i}, \qquad (b\in \seq{m}_{-}).
 \end{equation}
 The weight of the mutation at $b\in \seq{m}_{+}$ is given by
 \begin{equation}
  \label{eq:wt+}
  \frac{q^{ \frac{1}{2} (2s_{b}-\sum_{i\sim b} s_{i})\cdot (2s_{b}-0)}
  }{(q)_{k_{b}}}
  = 
  \frac{q^{\sum_{i=1}^{n} c_{ib}
  s_{i}s_{b}}}{(q)_{k_{b}}}.
 \end{equation}
 Clearly the relations \eqref{eq:ZI-ks-1} and \eqref{eq:ZI-ks-2} are
 combined into 
 \begin{equation}
  \label{eq:ks-rel-Dynkin}
  \seq{k}=C\seq{s},
 \end{equation}
 where $C$ is the Cartan matrix of type $\underline{Q}$.  

 Collecting \eqref{eq:wt-} and \eqref{eq:wt+}, the mutation weight of
 $\gamma$ is expressed as
 \begin{equation}
  \label{eq:Z(m)-S-summand} W(\gamma)=
   \prod_{a\in \seq{m}_{-}} \frac{q^{\sum_{i=1}^{n} c_{ia}
   s_{i}s_{a}}}{(q)_{k_{a}}} \prod_{b\in \seq{m}_{+}}
   \frac{q^{\sum_{i=1}^{n} c_{ib} s_{i}s_{b}}}{(q)_{k_{b}}}
   = \frac{q^{\sum_{i=1}^{n} c_{ij} s_{i}s_{j}}}
   {\prod_{i=1}^{n} (q)_{k_{i}}}
   =
   \frac{q^{C[\seq{s}]}}{(q)_{\seq{k}}},
 \end{equation}
 where we used $c_{ij}=c_{ji}$.  Since
 $\seq{s}=C^{-1}\seq{k}=D\seq{k}$, we have
 \begin{equation*}
  C[\seq{s}]= C[D\seq{k}] = \seq{k}^{T} (D^{T}CD) \seq{k}=
   \seq{k}^{T} D \seq{k}= D[\seq{k}].
 \end{equation*}
 Putting this into \eqref{eq:Z(m)-S-summand} and summing over $\seq{k}$,
 we obtain the desired formula for the partition $q$-series.
 \end{proof}

\section{Partition $q$-series and Characters II ---  square products}
\label{sec:ZII}

\subsection{Products of quivers and their mutations} 

\begin{figure}[b]
 \begin{equation*}
  \vcenter{
   \xymatrix@R=4.9mm @C=4.9mm @M=3pt{
    & \bullet \ar[r] &  \bullet &  \bullet \ar[l]\ar[r] & 
   \bullet & \bullet \ar[l] & \llap{$Q$} \\
  \bullet \ar[d] & \circ \ar[r]\ar[d] &  \circ \ar[d] &  
   \circ \ar[l]\ar[d]\ar[r] & \circ \ar[d] & \circ \ar[l]\ar[d]  \\ 
  \bullet & \circ \ar[r] &  \circ &  \circ \ar[l]\ar[r] & \circ & \circ
   \ar[l] \\
  \bullet \ar[u] & \circ \ar[r]\ar[u] &  \circ \ar[u] &  \circ
   \ar[l]\ar[u]\ar[r] & \circ \ar[u] & \circ \ar[l]\ar[u]  \\ 
  Q' &&& {\hbox to 0pt{$Q\otimes Q'$}}\\ 
   }}
   \qquad\quad
   \vcenter{
   \def\p{{\oplus}}
   \def\m{{\ominus}}
   \xymatrix@R=4mm @C=4mm @M=3pt{
  &   \bullet \ar[r] &  \bullet &  \bullet \ar[l]\ar[r] & 
   \bullet & \bullet \ar[l] & \llap{$Q$} \\
  \bullet \ar[d] & \p \ar[r] &  \m \ar[d] &  \p \ar[l]\ar[r] & 
   \m \ar[d] & \p \ar[l]  \\ 
  \bullet & \m \ar[d]\ar[u] &  \p \ar[l]\ar[r] &
   \m \ar[d]\ar[u] &  \p \ar[l]\ar[r] &
   \m \ar[d]\ar[u]  \\ 
  \bullet \ar[u] & \p \ar[r] &  \m \ar[u] &  \p \ar[l]\ar[r] &
   \m \ar[u] & \p \ar[l] 
   \\ 
  Q' &&& {\hbox to 0pt{$Q\square Q'$}}
   }}
 \end{equation*}
 \caption{Tensor product and square product of quivers.}\label{fig:QtimesQ'}
\end{figure}
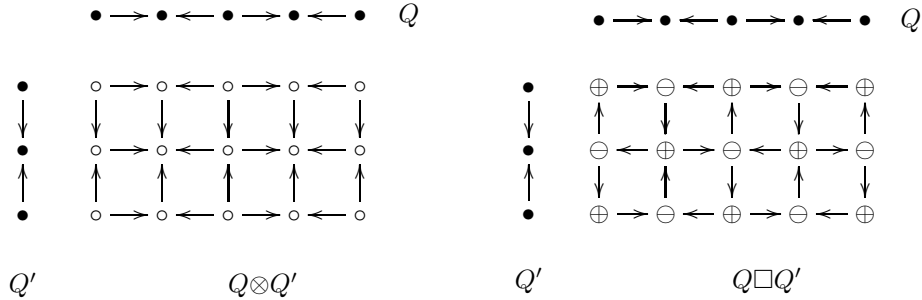

Let $Q, Q'$ be two quivers without oriented cycles, and
$B=(b_{ij}),B'=(b'_{i'j'})$ be the corresponding matrices.  The
\emph{tensor product $Q\otimes Q'$} is defined as follows
\cite{Keller2013} (Figure \ref{fig:QtimesQ'}): the set vertices is the
product $Q_0\times Q_0'$, and the associated matrix is given
by\footnote{There is a natural isomorphism between path algebras
$\C(Q\otimes Q')\simeq \C Q\otimes \C Q'$.}
\begin{equation}
 B(Q\otimes Q')=B(Q)\otimes I_{Q'} + I_{Q} \otimes B(Q')
\end{equation}
where $I_{Q}$, $I_{Q'}$ denotes the identity matrix of size $\# Q_{0}$,
$\#Q'_{0}$, respectively. 
In other words, the number of arrows from a
vertex $(i,i')$ to a vertex $(j,j')$
\begin{itemize}
\item[a)] is zero if $i\neq j$ and $i'\neq j'$;
\item[b)] equals the number of arrows from $j$ to $j'$ if $i=i'$;
\item[c)] equals the number of arrows from $i$ to $i'$ if $j=j'$.
\end{itemize}

Now assume that $Q$ and $Q'$ are alternating, i.e. each vertex is a
source or a sink.  We define the \emph{square product $Q\square Q'$} to
be the quiver obtained from $Q\otimes Q'$ by reversing all arrows in the
full subquivers of the form $\{i\}\times Q'$ and $Q\times \{i'\}$, where
$i$ is a source of $Q$ and $i'$ a sink of $Q'$.\footnote{This orientation
convention is slightly different from \cite{Keller2013}.}  Note that
$Q\otimes Q'$ has no oriented cycles, whereas $Q\square Q'$ is composed
of squares with oriented 4-cycle boundaries.  It is easy to check that
\begin{equation}
 (Q\otimes Q')^{op}= (Q^{op}\otimes Q'^{op}), \qquad 
   (Q^{op}\square Q')= (Q\square Q'^{op})= (Q\square Q')^{op}. 
\end{equation}

In the remainder of this section, we assume $Q$ and $Q'$ are alternating
quivers whose underlying graphs are of Dynkin diagram of ADE type.  The
vertices of $Q\square Q'$ are partitioned into two subsets: for
$\varepsilon = \pm $, we put
\begin{equation}
 \seq{m}_{\varepsilon} :=\bigl\{\,(i,i')\in Q_{0}\times Q'_{0}~|~
  \sgn(i)\sgn(i')=\varepsilon\bigr\}.
\end{equation}
In Figure \ref{fig:QtimesQ'}, $\seq{m}_{+}$ and $\seq{m}_{-}$
corresponds to vertices $\oplus$ and $\ominus$, respectively.  For each
$\varepsilon$, there is no arrows joining two vertices $v$, $v'$ of
$\seq{m}_{\varepsilon}$ and thus $\mu_{v}\circ \mu_{v'}=\mu_{v'}\circ
\mu_{v}$. 

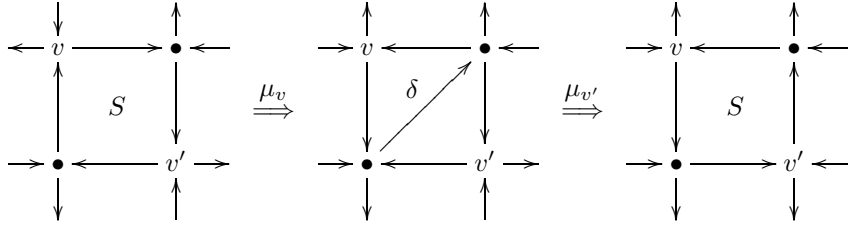
\begin{figure}[bpt]
\begin{equation*}
 \vcenter{
  \xymatrix @R=13pt @C=13pt @M=3pt{
  & \ar[d] & &&\\
  & v \ar[rr] \ar[l]&& \bullet \ar[dd]\ar[u] & \ar[l]\\
 && && \\
 \ar[r]&  \bullet  \ar[uu]\ar[d]
  \ar@{}[uurr]|{\textstyle S} && v' \ar[ll]\ar[r]& \\
  &&&  \ar[u] & \\ 
  }}
   \stackrel{\displaystyle \strut\mu_{v}}{\Longrightarrow}
 \vcenter{
 \xymatrix @R=13pt @C=13pt @M=3pt{
  & & &&\\
 \ar[r]  & v \ar[dd] \ar[u]&& \bullet  \ar[ll]\ar[dd]\ar[u] & \ar[l]\\
 &&&& \\
 \ar[r]&  \bullet  \ar[d]\ar[uurr]^{\textstyle \delta}&& v' \ar[ll]\ar[r]& \\
  & && \ar[u] & \\ 
  }}
   \stackrel{\displaystyle \strut\mu_{v'}}{\Longrightarrow}
 \vcenter{
 \xymatrix @R=13pt @C=13pt @M=3pt{
  & & &&\\
 \ar[r]  & v \ar[dd] \ar[u]&& \bullet  \ar[ll]\ar[u] & \ar[l]\\
 &&&& \\
 \ar[r]&  \bullet  \ar[d]\ar[rr]
  \ar@{}[uurr]|{\textstyle S} 
  && v' \ar[uu]\ar[d]& \ar[l]  \\
  & &&& \\ 
  }}
\end{equation*}
\caption{Creation and annihilation of a diagonal edge.} \label{fig:square-vv'}
\end{figure}

The following simple observation will be helpful.  Choose a square $S$
in $Q\square Q'$ and let $v$, $v'$ be the two vertices of $S$ in the
diagonal position (see Figure \ref{fig:square-vv'}). Suppose we perform
two mutations, first at $v$, and later at $v'$. By the mutation rule 2),
the first mutation creates an arrow $\delta$ which is a diagonal of
$S$. But the second mutation at $v'$ eliminates $\delta$ by the mutation
rule 3). As a result of combined mutation $\mu_{v}$ and $\mu_{v'}$, the
diagonal edge $\delta$ disappears, and the orientations of all arrows
bounding $S$ are reversed. Mutations on vertices other than $v$, $v'$
can never create or delete $\delta$.

With this observation in mind, it is easy to check that
\begin{equation*}
 \mu_{\seq{m}_{+}}(Q\square Q')=(Q\square Q')^{op},\qquad
  \mu_{\seq{m}_{-}}\left((Q\square Q')^{op}\right)=Q\square Q'.
\end{equation*}
Consequently, $\gamma = (Q\square Q';\seq{m}_{+}
\seq{m}_{-},\id)$ forms a mutation loop.

\subsection{Partition $q$-series}

We now consider the partition $q$-series for the mutation loop $\gamma =
(Q;\seq{m}_{+} \seq{m}_{-},\id)$.

As in Section \ref{sec:ZI}, every vertex $v$ of $Q\square Q'$ is mutated
exactly once; the $s$-variables before and after the mutation at $v$ are
identified $s_{v}=s'_{v}$ by the boundary condition. Both $s$- and
$k$-variables are thus in one to one correspondence with the vertex set
$Q_{0}\times Q'_{0}$; let $s_{(i,i')}$, $k_{(i,i')}$ be the $s$-,
$k$-variable associated with the vertex $(i,i')$, respectively. It is
useful to regard $\seq{s}=(s_{(i,i')})$ and $\seq{k}=(k_{(i,i')})$ as
column vectors indexed by the set $Q_{0}\times Q'_{0}$; we will use
lexicographic ordering, if necessary.

The main result of this section is the next
\begin{theorem}\label{thm:Z(Q,Q')-multisum} Let $Q$, $Q'$ be alternating
 quivers of type $A_{n}$, $D_{n}$ or $E_{n}$ with Cartan matrices
 $C_{Q}$, $C_{Q'}$, respectively.  Let $ \gamma = (Q\square Q'; \seq{m}_{+}
 \seq{m}_{-},\id)$ be the mutation loop described above. Then the
 partition $q$-series $Z(\gamma)$ has a following form:
 \begin{equation}
  \label{eq:Z(m-square)-S} Z(\gamma) =\sum_{\seq{k}\geq 0}
   \frac{q^{\frac{1}{2}\left( C_{Q}\otimes
   C_{Q'}^{-1}\right)[\seq{k}]}}{ (q)_{\seq{k}}}.
 \end{equation}
 The $s$- and
 $k$-variables are related as
 \begin{equation}
  \seq{k}=(I_{Q}\otimes C_{Q'})\seq{s}, \qquad \seq{s}=(I_{Q}\otimes
  C_{Q'}^{-1})\seq{k},
 \end{equation}
 where $I_{Q}$ is the identity matrix of size $\#Q_{0}$. 
\end{theorem}

\begin{rem}
 The following remarks are in order. 

 The partition $q$-series for the case when $Q$ is type $X$ and $Q'$ is
 type $A_{r-1}$ is of particular interest. Let $L(r\Lambda_{0})$ be the
 vacuum integrable highest weight module of the untwisted affine Lie
 algebra of type $X^{(1)}$.  The structures of various subquotients of
 this module, especially explicit description of their basis, are of
 considerable interest from the viewpoint of mathematical physics, and
 have been extensively studied \cite{Lepowsky1985, Kuniba1993,
 Feigin1993, Terhoeven, Georgiev1, Georgiev2, Stoyanovsky1994,
 Hatayama1998}. The corresponding characters are often referred to as
 fermionic formula or quasi-particle formula.  Precisely the same
 formula as \eqref{eq:Z(m-square)-S} appears in the literature (see for
 example (9) of \cite{Kuniba1993}, (0.5) of \cite{Georgiev2}, or (5.40)
 of \cite{Hatayama1998}). The relation with string functions
 \cite{Kac1984} reveals that when multiplied by a suitable power of $q$,
 $q^{s}Z(\gamma)$ becomes a modular form of some congruence subgroup of
 $SL_{2}(\Z)$.
\end{rem}

Before giving a proof, we illustrate the statement of Theorem
\ref{thm:Z(Q,Q')-multisum} using an example of $A_{3} \square A_{2}$:
\begin{equation*}
A_{3} \square A_{2}=\vcenter{
 \xymatrix @R=6mm @C=6mm @M=4pt{
 1 \ar[r] & 3  \ar[d] & 5\ar[l]\\
 2  \ar[u]& 4 \ar[l] \ar[r]& 6 \ar[u]}
}.
\end{equation*}
Here we enumerate the vertices in the lexicographical order:
\begin{equation*}
 1\leftrightarrow (1,1), \quad2\leftrightarrow (1,2),
 \quad3\leftrightarrow (2,1), \quad4\leftrightarrow (2,2),
 \quad5\leftrightarrow (3,1), \quad6\leftrightarrow (3,2).
\end{equation*}
We consider the mutation loop (see Figure \ref{fig:QQ'})
\begin{equation*}
 \seq{m}_{+}= (1,4,5),\qquad \seq{m}_{-}= (2,3,6), \qquad
  \gamma=(A_{3} \square A_{2};(1, 4, 5, 2, 3, 6),\id).
\end{equation*}

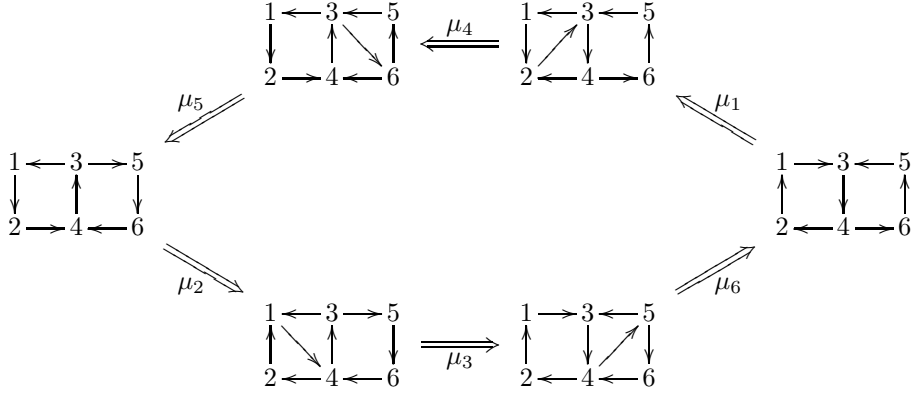
\begin{figure}[bt]
\begin{xy}
 \xygraph{!{<0cm,0cm>;<17mm,0cm>:<0cm,20mm>::}
  !{(3,0)}*++{\vcenter{\xybox{
 \xymatrix @R=5mm @C=5mm @M=2pt{
  1 \ar[r] & 3  \ar[d] & 5\ar[l]\\
  2  \ar[u]& 4 \ar[l] \ar[r]& 6 \ar[u]}
}}}="Q0"
  !{(1,1)}*++{\vcenter{\xybox{
 \xymatrix @R=5mm @C=5mm @M=2pt{
  1 \ar[d] & 3  \ar[d]\ar[l] & 5\ar[l]\\
  2  \ar[ur]& 4 \ar[l] \ar[r]& 6 \ar[u]}
}}}="Q1"
  !{(-1,1)}*++{\vcenter{\xybox{
\xymatrix @R=5mm @C=5mm @M=2pt{
  1 \ar[d] & 3 \ar[l] \ar[dr] & 5\ar[l]\\
  2  \ar[r]& 4 \ar[u] & 6 \ar[u]\ar[l]}
}}}="Q2"
  !{(-3,0)}*++{\vcenter{\xybox{
\xymatrix @R=5mm @C=5mm @M=2pt{
  1 \ar[d] & 3  \ar[l]  \ar[r] & 5\ar[d]\\
  2  \ar[r]& 4 \ar[u] & 6 \ar[l]}
}}}="Q3"
  !{(-1,-1)}*++{\vcenter{\xybox{
\xymatrix @R=5mm @C=5mm @M=2pt{
  1 \ar[dr] & 3  \ar[l]\ar[r] & 5\ar[d]\\
  2  \ar[u]& 4 \ar[l] \ar[u]& 6 \ar[l]}
}}}="Q4"
  !{(1,-1)}*++{\vcenter{\xybox{
\xymatrix @R=5mm @C=5mm @M=2pt{
  1 \ar[r] & 3  \ar[d] & 5\ar[l]\ar[d]\\
  2  \ar[u]& 4 \ar[l] \ar[ru]& 6 \ar[l]}
}}}="Q5"
"Q0":@{=>}_{\textstyle \mu_{1}}"Q1"
"Q1":@{=>}_{\textstyle \mu_{4}}"Q2"
"Q2":@{=>}_{\textstyle \mu_{5}}"Q3"
"Q3":@{=>}_{\textstyle \mu_{2}}"Q4"
"Q4":@{=>}_{\textstyle \mu_{3}}"Q5"
"Q5":@{=>}_{\textstyle \mu_{6}}"Q0"
}
\end{xy}
 \caption{The mutation loop $\gamma=(A_{3}\square
 A_{2};\seq{m}_{+}\seq{m}_{-},\id)$. Here $\seq{m}_{+}=(1,4,5)$,
 $\seq{m}_{-}=(2,3,6)$.}\label{fig:QQ'}
\end{figure}

By the boundary condition, $s$-variables before and after mutation are
identified vertex-wise.  The linear relations between $k$- and
$s$-variables are
\begin{equation}
\begin{split}
&  k_1=2 s_1-s_2,\quad k_4=2 s_4-s_3,\quad k_5=2 s_5-s_6,
\\
& k_2=2 s_2-s_1,\quad k_3=2 s_3-s_4,\quad k_6=2  s_6-s_5.
\end{split}
\end{equation}
This may be written as
\begin{equation*}
 \arraycolsep=2pt
  \left(
   \begin{array}{c}
    k_1 \\
    k_2 \\
    k_3 \\
    k_4 \\
    k_5 \\
    k_6
   \end{array}
       \right)
  =
 \left(
\begin{array}{rrrrrr}
 2 & -1 & 0 & 0 & 0 & 0 \\
 -1 & 2 & 0 & 0 & 0 & 0 \\
 0 & 0 & 2 & -1 & 0 & 0 \\
 0 & 0 & -1 & 2 & 0 & 0 \\
 0 & 0 & 0 & 0 & 2 & -1 \\
 0 & 0 & 0 & 0 & -1 & 2
\end{array}
\right)
\left(
\begin{array}{c}
 s_1 \\
 s_2 \\
 s_3 \\
 s_4 \\
 s_5 \\
 s_6
\end{array}
\right)
\end{equation*}
or more compactly,
\begin{equation}
 \label{eq:ks-QQ'}
 \seq{k}= \left( I_{3}\otimes C_{A_{2}}\right) \seq{s}.
\end{equation}
The weight is given by
\begin{equation}
 \label{eq:W-QQ'}
 \begin{split}
  W(\gamma)& =
  \frac{q^{\frac{1}{2} \left(2 s_1-s_2\right) 
  \left(2 s_1-s_3\right)}}{(q)_{2 s_1-s_2}} 
  \frac{q^{\frac{1}{2} \left(2 s_4-s_3\right) 
  \left(2 s_4-s_2-s_6\right)}}{(q)_{2 s_4-s_3}} 
  \frac{q^{\frac{1}{2} \left(2 s_5-s_3\right)
  \left(2 s_5-s_6\right)}}{(q)_{2 s_5-s_6}} 
  \\
  &\quad \times 
  \frac{q^{\frac{1}{2} \left(2 s_2-s_1\right) \left(2
  s_2-s_4\right)}}{(q)_{2 s_2-s_1}} 
  \frac{q^{\frac{1}{2}
   \left(2 s_3-s_4\right) \left(2 s_3-s_1-s_5\right)}}{(q)_{2 s_3-s_4}} 
  \frac{q^{\frac{1}{2} \left(2 s_6-s_4\right) 
  \left(2  s_6-s_5\right)}}{(q)_{2 s_6-s_5}}.
\end{split}
\end{equation}
The numerator of \eqref{eq:W-QQ'} is of the form
$q^{\frac{1}{2}C[\seq{s}]}$, where $C[\seq{s}]$ is a quadratic form
defined by the positive definite symmetric matrix
\begin{equation*}
C= 
\arraycolsep=2pt
\left(
\begin{array}{rrrrrr}
 4 & -2 & -2 & 1 & 0 & 0 \\
 -2 & 4 & 1 & -2 & 0 & 0 \\
 -2 & 1 & 4 & -2 & -2 & 1 \\
 1 & -2 & -2 & 4 & 1 & -2 \\
 0 & 0 & -2 & 1 & 4 & -2 \\
 0 & 0 & 1 & -2 & -2 & 4
\end{array}
\right)
=
\left(
\begin{array}{rrr}
 2 & -1 & 0 \\
 -1 & 2 & -1 \\
 0 & -1 & 2
\end{array}
\right)
\otimes 
\left(
\begin{array}{rr}
 2 & -1 \\
 -1 & 2
\end{array}
\right)
=C_{A_{3}}\otimes C_{A_{2}}.
\end{equation*}
We have $\seq{s}= (I_{3}\otimes C_{A_{2}}^{-1}) \seq{k}$ by inverting
the relation \eqref{eq:ks-QQ'}. Substituting this into \eqref{eq:W-QQ'},
we obtain the partition $q$-series:
\begin{equation*}
  Z(\gamma) = \sum_{k_{1},k_{2},k_{3},k_{4},k_{5},k_{6}=0}^{\infty}
   \frac{q^{\frac{1}{2}D[\seq{k}] }}{(q)_{k_1} (q)_{k_2} (q)_{k_3}
   (q)_{k_4} (q)_{k_5} (q)_{k_6}}
\end{equation*}
where
\begin{equation*}
\arraycolsep=2pt
 D=\frac{1}{3}
\left(
\begin{array}{rrrrrr}
 4 & 2 & -2 & -1 & 0 & 0 \\
 2 & 4 & -1 & -2 & 0 & 0 \\
 -2 & -1 & 4 & 2 & -2 & -1 \\
 -1 & -2 & 2 & 4 & -1 & -2 \\
 0 & 0 & -2 & -1 & 4 & 2 \\
 0 & 0 & -1 & -2 & 2 & 4
\end{array}
\right)
= 
\left(
\begin{array}{rrrrrr}
 2 & -1 & 0 \\
 -1 & 2 & -1 \\
 0 & -1 & 2
\end{array}
\right)
\otimes 
\left(
{\arraycolsep=4pt
\begin{array}{cc}
 \frac{2}{3} & \frac{1}{3} \\[5pt]
 \frac{1}{3} & \frac{2}{3}
\end{array}
}
\right)
= C_{A_{3}}\otimes (C_{A_{2}})^{-1}.
\end{equation*}

\begin{proof} of Theorem \ref{thm:Z(Q,Q')-multisum}:

First consider the sequence of mutations $\seq{m}_{+}$ applied to
$Q\square Q'$.  Pick a vertex $v= (i,i') \in \seq{m}_{+}$ (marked with
$\oplus$ in Figure \ref{fig:QtimesQ'}.) Then by the very definition of
$Q\square Q'$, every incoming arrow $\alpha$ to $v$ comes from
``vertical'' directions; $\alpha$ is of the form $(i,j')\to (i,i')$
where $j'$ is adjacent to $i'$ in the underlying graph
$\underline{Q}'$. This means that the $k$- and $s$-variables are
related as
\begin{equation}
 \label{eq:ks-rel-QQ}
 k_{(i,i')} = 2s_{(i,i')}-\sum_{i'\sim j' \in \underline{Q}'} s_{(i,j')}
\end{equation}

Next we take up the mutation sequence $\seq{m}_{-}$ (marked with
$\ominus$ in Figure \ref{fig:QtimesQ'}). Since the mutation
$\seq{m}_{-}$ is applied only after $\seq{m}_{+}$ is over, it is
convenient to consider $\mu_{\seq{m}_{+}}(Q\square Q')=(Q\square
Q')^{op}$ as the initial quiver. Then the connections around the
mutating vertex is exactly the same as before: all incoming arrows again
come from ``vertical'' directions. Thus, \eqref{eq:ks-rel-QQ} is true
for $v=(i,i')\in \seq{m}_{-}$ as well.  Thus we have the relation
\begin{equation}
 \label{eq:QQks}
 \seq{k} = (I_{Q}\otimes C_{Q'}) \seq{s}. 
\end{equation}
Since $Q'$ is of $ADE$ type, $C_{Q'}$ is a positive definite symmetric
matrix. Thus the linear relation \eqref{eq:QQks} is invertible:
\begin{equation}
  \seq{s} = (I_{Q}\otimes C_{Q'}^{-1}) \seq{k}. 
\end{equation}
In particular, the mutation loop $\gamma$ is nondegenerate.

We have seen that all incoming arrows to the mutating vertices run
``vertically.''  This means that all outgoing arrows run
``horizontally.'' Therefore, $k^{\vee}$-variables, which are introduced
in \eqref{eq:kv-and-s}, are related with $s$-variables as
\begin{equation}
 \label{eq:kvs-rel-QQ}
 k^{\vee}_{(i,i')} = 2s_{(i,i')}-\sum_{i \sim j \in \underline{Q}} s_{(j,i')},
\end{equation}
or equivalently, 
\begin{equation}
 \label{eq:kvs-rel} 
 \seq{k}^{\vee} = (C_{Q}\otimes I_{Q'}) \seq{s}. 
\end{equation}
The weight of the whole mutation sequence is then
\begin{equation}
 \label{eq:W(m-square)}
 W(\gamma)=\prod_{v}\frac{q^{\frac{1}{2}k_{v}k_{v}^{\vee}}}{(q)_{k_{v}}}
  = \frac{q^{\frac{1}{2}\sum_{v} k_{v}k_{v}^{\vee}}}{(q)_{\seq{k}}}.
\end{equation} 
Note that the sum in the numerator is written as 
\begin{equation}
 \label{eq:QQ'qexp2}
 \begin{split}
  \sum_{\smash{v}}k_{v}k_{v}^{\vee} &= \seq{k}^{T} \seq{k}^{\vee} 
   = \seq{k}^{T} \,
   (C_{Q}\otimes I_{Q'}) \seq{s}
  \\[-1em]
  & = 
  \seq{k}^{T}\,  (C_{Q}\otimes I_{Q'}) (I_{Q}\otimes C_{Q'}^{-1})\seq{k}
  = 
  \seq{k}^{T}\, (C_{Q}\otimes C_{Q'}^{-1}) \seq{k}
  \\
  & = (C_{Q}\otimes C_{Q'}^{-1}) [\seq{k}].
 \end{split}
\end{equation}
Combining \eqref{eq:W(m-square)}, \eqref{eq:QQ'qexp2},
and summing over the $k$-variables, we obtain the desired formula
\eqref{eq:Z(m-square)-S}.

\end{proof}
\begin{rem}
 With the same initial quiver $Q\square Q'$, we can construct another
 mutation loop $\gamma'=(Q\square Q'; \seq{m}_{-}\seq{m}_{+},\id)$ by exchanging
 $\seq{m}_{+}$ and $\seq{m}_{-}$. Analysis similar to that in the proof of
 Theorem \ref{thm:Z(Q,Q')-multisum} show that $Q$ and $Q'$ exchange
 their roles; the partition $q$-series is now given by
 \begin{equation}
  \label{eq:Z(gamma')} Z(\gamma')
   =\sum_{\seq{k}\geq 0} \frac
  {q^{\frac{1}{2}\left( C_{Q}^{-1}\otimes C_{Q'}\right)[\seq{k}]}}{
  (q)_{\seq{k}}}.
 \end{equation}

\end{rem}

%\bibliographystyle{abbrv}
%\bibliography{mutation}

\end{document}